\documentclass[11pt]{article}

\usepackage{sn-preamble} 
\usepackage{tikz}
\usepackage{verbatim}
\usepackage{cancel}
\usepackage{enumitem}
\usepackage{bm}
\usepackage{caption}
\usepackage{changepage}   
\usepackage{graphicx,subcaption}

\newcommand{\ignore}[1]{}

\newcommand{\dpos}{\calD_{+}}
\newcommand{\dneg}{\calD_{-}}
\newcommand{\good}{\text{good}}
\newcommand{\bad}{\text{bad}}
\newcommand{\hgood}{\calH_+^{\good}}

\newcommand{\dheavy}{\calD_-^{(1)}}

\newcommand{\sgn}{{\mathrm{sign}}}

\newcommand{\nmarg}{\eta}
\newcommand{\findh}{\textsc{FindGoodHalfspace}}

\newcommand{\cover}{\textsc{CoverLearner}}

\title{Tight Bounds for Learning Polyhedra with a Margin}
\author{Shyamal Patel\thanks{Email: \texttt{shyamalpatelb@gmail.com}} \\ \textsl{Columbia University} \and Santosh S. Vempala\thanks{Email: \texttt{vempala@gatech.edu}}\\ \textsl{Georgia Tech}}
\date{}

\begin{document}

\maketitle

\begin{abstract}
We give an algorithm for PAC learning intersections of $k$ halfspaces with a $\rho$ margin to within error $\eps$ that runs in time $\poly(k, \eps^{-1}, \rho^{-1}) \cdot \mbox{exp}(O(\sqrt{n \log(1/\rho) \log k}))$. Notably, this improves on prior work which had an exponential dependence on either $k$ or $\rho^{-1}$ and matches known cryptographic and Statistical Query lower bounds up to the logarithmic factors in $k$ and $\rho$ in the exponent. Our learning algorithm extends to the more general setting when we are only promised that most points have distance at least $\rho$ from the boundary of the polyhedron, making it applicable to continuous distributions as well. 
\end{abstract}

\section{Introduction}
Over the past half-century of computational learning theory, the problem of learning a single halfspace has emerged as its  archetypal ``well-solvable" problem and its study has led to the development of new and general techniques for its efficient solution in a variety of settings (supervised, unsupervised, semi-supervised, statistical query, agnostic,...). That said, provable guarantees for natural extensions of halfspaces have been a long-standing challenge. One such generalization of great interest is the problem of learning an intersection of finitely many halfspaces: the labeling function $f:\R^n \rightarrow \{-1, 1\}$ is defined by an unknown set of $k$ halfspaces: for $x \in \R^n$, 

\[
f(x) = \begin{cases}
1 \quad \mbox{ if } w_i \cdot x > \theta_i \mbox{ for all } i\in [k]\\
-1 \quad \mbox{otherwise.}
\end{cases}
\]

When $k = \poly(n)$, this problem captures many well-known learning problems of independent interest, such as learning Disjunctive Normal Forms (DNFs) and learning polyhedra. Even for $k=2$ halfspaces, only recently \cite{alman2026learning} gave a subexponential algorithm for PAC learning, with a doubly exponential dependence on $k$. One typically imposes additional restrictions to make the problem tractable. Most relevant to this work is the assumption of a $\rho$ margin, which roughly corresponds to each point having distance at least $\rho$ from any of the $k$ halfspaces defining the target function (cf. \Cref{def:margin} and \Cref{def:robust}). In this setting, \cite{klivans2008margin} give an algorithm whose complexity grows as either $(k/\rho)^{\wt{O}(k\log(1/\rho))}$ or $(\log k/\rho)^{\wt{O}(\sqrt{1/\rho}\log k)}$, improving on earlier work of~\cite{Arriaga2006algorithmic} and drawing on techniques from \cite{klivans2004intersections}. For small $\rho$, this was improved by Gottfried at al.~\cite{gottlieb2021learning}\footnote{As we discuss in \Cref{sec:margin}, this is a different but related notion of margin compared to \cite{klivans2008margin} and \cite{Arriaga2006algorithmic}.} who showed that there is a quasi-proper learning algorithm that uses $m = \wt{O}(\rho^{-1} \cdot k)$ examples and runs in time $m^{\wt{O}(1 / \rho^2)}$. Another improvement was given by \cite{goel2019learning}, who proved that an intersection of $k$ halfspaces can be learned in time $\poly(n, k^{O(1/\rho)})$. For learning polyhedra with $k = \poly(n)$ facets and $\rho = 1/\poly(n)$ margin, these bounds are all exponential in the dimension.

The complexity of learning an intersection of halfspaces for {\em structured} input distributions has also been a subject of intensive study. Following early work by Baum~\cite{Baum90, Baum90b}, Blum and Kannan~\cite{blum1997} gave an algorithm for learning an intersection of a fixed number of halfspaces from the uniform distribution on the unit ball. This was later improved in ~\cite{Vem10a} to complexity $n^{O(k)}$ for logconcave distributions (with halfspaces going through the origin). 
When the input distribution is Gaussian, Klivans, O'Donnell and Servedio~\cite{KOS08} gave algorithms with complexity  $2^{O(\sqrt{n}/\eps^4)}$ (independent of $k$) and $n^{\wt{O}(\log k/\varepsilon^4)}$ for learning to within error $\eps$. The latter bound was improved in~\cite{Vem10b} to complexity $\poly(n,k) \cdot \min \{k^{\wt{O}(\log k/\varepsilon^4)}, (k/\varepsilon)^{O(k)}\}$, which allows for learning an intersection of up to $k=2^{O(\sqrt{\log n})}$ halfspaces in polytime. Additionally, under the uniform distribution over the Boolean cube, \cite{kane2014average} showed an algorithm with running time $n^{O(\log(k)/\eps^2)}$.
Notably, for structured distributions beyond uniform and Gaussian, when the number of halfspaces is polynomial in the dimension (as is the case for learning polyhedra), the best known time complexity is again exponential in the dimension.

\paragraph{Lower bounds.} 
It is natural to wonder why techniques that work so well for halfspaces, e.g., Perceptron-style algorithms, linear programming, outlier removal, isotropic transformation, polynomial approximation, and $\ell_1$-minimization do not extend to the intersection of halfspaces. First, proper learning is provably hard unless P=NP, even for $k=2$ halfspaces~\cite{BR92, alekhnovich2004learnability, khot2008hardness}. Next, when $k$ is part of the input, the same paper shows that learning an intersection of $k$ halfspaces as an intersection of at most $kn^{1-\eps}$ halfspaces is also hard via a reduction from hypergraph coloring~\cite{BR92}. \cite{gottlieb2021learning} proved that under the Exponential Time Hypothesis (ETH), proper learning of an intersections of $k$ halfspaces with a margin requires $\exp((k/16\gamma^2)^{1-\delta})$ time. Of course, these results do not rule out efficient learning, as the hypothesis class used could be any function class that can be learned in polytime. One such natural class is bounded-degree polynomials, which are used for learning various hypotheses such as De Morgan Formulas and Disjunctive Normal Forms (DNFs) \cite{klivans2004dnf, lee2009note}. Unfortunately, approximating an intersection of two halfspaces with a polynomial threshold function needs a polynomial of degree $\Omega(n)$~\cite{sherstov2013optimal}; since learning polynomial-threshold functions (PTFs) requires time exponential in the degree of the polynomial, the standard approach to learning intersections of halfspaces as a PTF requires $2^{\Omega(n)}$ time. 

From a broader perspective, state-of-the-art algorithms for learning halfspaces and low-degree polynomials can be viewed as {\em Statistical Query} algorithms, a model of learning introduced by Kearns~\cite{kearns1998sq} and extended and applied widely in subsequent decades~\cite{blum1994weakly,feldman2017statistical}. For many well-known problems in supervised and unsupervised learning, the best-known upper bounds are achievable by SQ algorithms and have nearly matching lower bounds in the SQ model of computation. For intersections of halfspaces,  Klivans and Sherstov~\cite{klivans2007unconditional} showed that learning an intersection of $\sqrt{n}$ halfspaces with $1/\poly(n)$ margin over the Boolean cube has SQ complexity $2^{\widetilde{\Omega}(\sqrt{n})}$. For small $k$, this was later improved by Tiegel \cite{tiegel2024improved}, and a minor modification of Tiegel's approach also yields a simple proof attaining the $2^{\Omega(\sqrt{n})}$ SQ lower bound for  intersections of $\sqrt{n}$ halfspaces (see Appendix~\ref{app:SQ}).

Another line of work has developed representation-independent lower bounds for this problem based on cryptographic assumptions. Klivans and Sherstov~\cite{klivans2009cryptographic} showed hardness of PAC-learning an intersection of $k=n^{\eps}$ halfspaces  for some $\eps > 0$ assuming that the unique Shortest Vector Problem in lattices is hard to approximate within a factor of $\wt{O}(n^{1.5})$.  This was improved by Tiegel~\cite{tiegel2024improved}, who showed a lower bound of $n^{\Omega(k/\log k)}$ for $k \le \sqrt{n}$ halfspaces under similar lattice problem hardness assumptions. 

\paragraph{Our work.} In this paper, we give a simple new algorithm for PAC-learning an intersection of halfspaces. The algorithm has complexity $2^{\wt{O}(\sqrt{n})}$ for any target function with $k\le \poly(n)$ halfspaces and  margin $1/\poly(n)$ (see Theorem~\ref{thm:main-hard-margin} for the precise complexity). It is based on a conceptually very simple weak learning algorithm: we output a random halfspace from the convex set of halfspaces that correctly classify a sufficiently large sample of positive examples and negative examples (notably of different sizes). 

We then refine the usual notion of margin to a weaker notion of ``soft'' margin $\bar{\rho}$, which only requires that the part of the negative distribution within distance $\bar{\rho}$ of the positive region (the target polyhedron) is at most a $1/\poly(n)$ fraction (see \Cref{def:soft-margin} for the precise definition). This allows us to extend our results to continuous distributions including ``hard" input distributions such as mixtures of Gaussians (``parallel pancakes") used in lower bounds for the problem, specifically in \cite{tiegel2024improved}, as these distributions do not have hard margins but do have $1/\poly(n)$ soft margins. The complexity of our algorithm for a $\rho$-soft margins is $2^{O(\sqrt{n \log (k\eps^{-1}) \log (1/\bar{\rho})})}$ to learn the distribution outside the margin to accuracy $1-\eps$ (see \Cref{thm:main-full}). Thus, our algorithm provides a nearly tight upper bound matching these lower bounds for learning intersections of $\poly(n)$ halfspaces.

\subsection{Results}
Our first result is the following theorem for learning an intersection of halfspaces with a margin.  

\begin{theorem}
	\label{thm:main-hard-margin} 
	Let $\calD$ be a distribution, $f$ an intersection of $k$ halfspaces, and $\eps,\delta, \rho \in (0,1/3)$. If $\calD$ has a $\rho$ margin with respect to $f$ (cf. \Cref{def:margin}), then there exists an algorithm that $(\eps,\delta)$-PAC learns $f$  in time
	 \[ 2^{O\left(\sqrt{n \log(1/\rho) \log k}\right)} \cdot \poly\left(k, \eps^{-1}, \rho^{-1}, \log(\delta^{-1})\right)\]
\end{theorem}

This improves on previous results that have an exponential dependence on either $k$ or $\rho^{-1}$, and so are exponential-time algorithms for learning polytopes with $\poly(n)$ facets and $1/\poly(n)$ margin \cite{klivans2004learning}. 

As an immediate consequence, we get the following.

\begin{corollary}
In the distribution-free PAC learning model, there is an algorithm to learn the class of intersections of $k$ halfspaces with integer weights of magnitude at most $W$ over the Boolean cube $\bits^n$ to error $\eps$ with probability $1 - \delta$ in time 
\[2^{O\left(\sqrt{n \log(nW) \log k}\right)} \cdot \poly\left(k, \eps^{-1}, nW, \log(\delta^{-1}) \right) \]
\end{corollary}

This follows from \Cref{thm:main-hard-margin}, as low-weight halfspaces have margin at least $\poly(n,W)^{-1}$ over any distribution over the Boolean cube. Notably, intersections of low weight halfspaces generalize DNFs, a central object of study in the PAC learning of Boolean functions. 

This connection also implies that \Cref{thm:main-hard-margin} cannot be improved by any SQ algorithm (up to the logarithmic factors in $\rho, n$ and $k$). Indeed, the SQ lower bound in \cite{klivans2007unconditional} shows that an intersection of $\sqrt{n}$ halfspaces with weights bounded by $\poly(n)$ requires $2^{\wt{\Omega}(\sqrt{n})}$ time to learn over the Boolean cube. As such, \Cref{thm:main-hard-margin} matches the SQ lower bound up to logarithmic factors in the exponent.

Strictly speaking, the requirement of a hard margin in \Cref{thm:main-hard-margin} does not allow it to be used for continuous distributions (including uniform, logconcave, Gaussian mixtures etc.). Using a weaker notion of margin, we prove the following ``soft" margin extension, where the soft margin roughly corresponds to the fraction of $\calD$ that is labeled negative and within distance $\rho$ of the region labeled positive by $f$.

\begin{theorem}
	\label{thm:main-full} 
	Let $\calD$ be a distribution, $f$ an intersection of $k$ halfspaces, and $\eps,\delta, \rho \in (0,1/3)$. Let $\nmarg=\nmarg(\rho,\calD,f)$ be the soft margin (cf. \Cref{def:soft-margin}). There exists an algorithm to PAC learn $f$ to  error at most $\nmarg + \eps$ with probability at least $1 - \delta$ in time
	\[ 2^{O\left(\sqrt{n \log(1/\rho) \log(k\eps^{-1})}\right)} \cdot \poly(k, \eps^{-1},  \rho^{-1}, \log(\delta^{-1})). \]
Moreover, the output hypothesis is an intersection of halfspaces.
\end{theorem}

We note that distributions with mild {\em anti-concentration} (such as logconcave distributions) would have a soft margin that depends only on the covariance of the distribution.  
This result provides a nearly matching upper bound for known cryptographic lower bounds~\cite{klivans2009cryptographic,tiegel2024improved}.
\subsection{Technical Overview}
\subsubsection{Overview of \Cref{thm:main-hard-margin}}
We start by discussing the proof of \Cref{thm:main-hard-margin}. Throughout this overview, we will assume without loss of generality that all halfspaces are origin centered, $\supp(\calD) \subseteq \mathbb{S}^{n-1}$, and denote the target function $f(x) = \bigwedge_{i=1}^k \sgn(w_i \cdot x)$ where $\|w_i\|_2=1$. We'll also assume for simplicity in this outline that $f(x)$ is balanced, i.e. $\Pr_{\bx \sim \calD}[f(\bx) = 1] = \frac{1}{2}$. Roughly speaking, this is without loss of generality for us as we will be using a boosting based approach. As such, if $f$ is noticeably biased, a constant function will constitute a weak learner.

We preface our proof techniques by noting that our methods differ significantly from prior work on learning an intersection of halfspaces with a margin \cite{klivans2004intersections, klivans2008margin}, which rely on showing that an intersection of halfspaces can be approximated by a low degree PTF. At a high-level, our techinques share some similarity with those of \cite{alman2026learning}, in that we will try to ``guess'' a good halfspace. That said, unlike \cite{alman2026learning}, we will not aim to find a region where all but one halfspace is fixed. Instead, our goal will be to find a halfspace $h$ that satisfies both
\begin{align*}
	&(i)  \mbox{ (Captures negatives)} \quad \Pr_{\bx \sim \calD}[h(\bx) = -1] \geq \poly(k^{-1},\rho^{-1}) \cdot 2^{-O \left(\sqrt{n \log(k) \log(1/\rho)} \right)}\\ \\
	&(ii) \mbox{ (Advantage)} \qquad \Pr_{\bx \sim \calD}[f(\bx) = -1 |h(\bx) = -1] \geq 0.51
\end{align*}

Notably, if we could find such a halfspace, then we could easily generate a weak learner that learns to accuracy $\frac{1}{2} + \gamma$ where 
	\[\gamma := 0.01 \cdot \poly(k^{-1},\rho^{-1}) \cdot 2^{-O \left(\sqrt{n \log(k) \log(1/\rho)} \right)}\]
by outputting $-1$ when $h(x) = -1$ and flipping a coin otherwise. Standard boosting results \cite{schapire1990strength, freund1995boosting} would then yield \Cref{thm:main-hard-margin}.

To compute such a halfspace $h$, we start by drawing a set $\bP$ of $M_+$ random examples, where $M_+$ is a parameter to be set later, from $\calD$ satisfying $f(\bx) = 1$ and consider all halfspaces that label $\bP$ positively, i.e.
		\[\calH_+ = \{ w \in B_2^n: w \cdot x \geq 0 \quad \forall x \in \bP\} \]
	where $B_2^n = \{w \in \R^n: \|w\| \leq 1\}$ denotes the Euclidean ball.
	
By standard VC bounds, we expect all halfspaces $w \in \calH_+$ to have a very small fraction of points with $f(x) = 1$ lying below $w$, i.e. 
	\[\Pr_{\bx \sim \calD} [w \cdot \bx < 0 | f(\bx) = 1] \leq \frac{100 n \log(M_+)}{M_+}.\]
While this seems promising for learning, such a $w$ could have $w \cdot x > 0$ for all $x \in \supp(\mathcal{D})$, which would render it useless.

Nevertheless, we'll focus our attention on ``good'' halfspaces, which we define to be the subset of $\calH_+$ with a nontrivial fraction of points below them. That is,
\begin{equation}
\label{eq:tech-good}
	\Pr_{\bx \sim \calD} [w \cdot \bx < 0 | f(\bx) = -1] \geq \frac{500 n \log(M_+)}{M_+}.
\end{equation}
Notably, any good halfspace will satisfy properties $(i)$ and $(ii)$. Indeed, $(i)$ follows immediately so long as $M_+ \leq \poly(k,\rho^{-1}) \cdot 2^{O \left(\sqrt{n \log(k) \log(1/\rho)} \right)}$. To see that $(ii)$ holds, note that
\begin{align*}
	\Pr[f(\bx) = -1 \land h(\bx) = -1] &= \Pr[h(\bx) = -1] - \Pr[f(\bx) = +1 \land h(\bx) = -1] \\
	& \geq \Pr[h(\bx) = -1] - \frac{100 n \log(M_+)}{M_+}.
\end{align*}
We can then get property $(ii)$ by combining the above with the fact that
\[\Pr_{\bx \sim \calD}[h(\bx) = -1] \geq \Pr[f(\bx) = -1] \cdot \Pr_{\bx \sim \calD} [w \cdot \bx < 0 | f(\bx) = -1] \ge \frac{250 n \log(M_+)}{M_+} \]
by \Cref{eq:tech-good} and the fact that $f$ is balanced.

Thus, it now suffices to describe how to find a good halfspace in $\calH_+$. Unfortunately, it is not hard to construct instances where most halfspaces in $\calH_+$ are bad, i.e. they do not satisfy \Cref{eq:tech-good}. To circumvent this, we will sample a set $\bN$ of $M_-$ examples and consider the set of halfspaces that negatively classify all points in $\bN$ (see \Cref{fig:margin_and_algo}):
	\[\calH = \{w \in \calH_+: w \cdot x < 0 \quad \forall x \in \bN\}\]
The attentive reader may now note that such a body could be empty as defined. Nevertheless, we will focus the analysis on the case when every sample in $\bN$ satisfies $w_1 \cdot x < 0$. We can assume without loss of generality that
	\[\Pr_{\bx \sim \calD}[ w_1 \cdot \bx < 0 | f(\bx) = -1] \geq \frac{1}{k}.\]
Under this assumption on $\bN$, we can easily see that $w_1 \in \calH$. Moreover, since every point has margin $\rho$ with respect to $w_1$, every halfspace with weights in $w_1 + \rho B_2^n$ will correctly classify $\bP$ and $\bN$.\footnote{We are assuming here that $|w_i \cdot x| \geq \rho$ for all positively labeled points $x$ and all $i \in [k]$. While this is implied by robustness (cf. \Cref{def:robust}), this is not guaranteed by our definition of margin in \Cref{def:margin}. That said, this  assumption will not significantly change the argument and we make it for simplicity.} Thus it follows that we have the stronger property that
	\[ w_1 + \rho B_2^n \subseteq (1 + \rho) \calH\]
(Note that we must scale $\calH$ as points in $w_1 + \rho B_2^n$ may have length $1 + \rho$.)

On the other hand, for a bad halfspace $w \in \calH_+$, we note that 
	\[\Pr_{\bx \sim \calD} \left[w \cdot \bx < 0 | w_1 \cdot \bx < 0 \right] \leq k \cdot \frac{500n\log(M_+)}{M_+}.\]
So the probability a bad halfspace survives $M_-$ such samples is bounded by
	\[\left( k \cdot \frac{500n\log(M_+)}{M_+} \right)^{M_-}\]
from which we conclude that the expected volume of the bad halfspaces is at most
	\[|\calH_+| \cdot \left( k \cdot \frac{500n\log(M_+)}{M_+} \right)^{M_-} \leq |B_2^n| \left( k \cdot \frac{500n\log(M_+)}{M_+} \right)^{M_-}.\]
Now if the volume of the bad halfspaces in $\calH$ is smaller than the volume of the good halfspaces in $\calH$, randomly sampling from $\calH$ would yield a good halfspace $h$ with good probability. As such, to sample a halfspace $h$, it suffices that
	\begin{equation}
	\label{eq:tech-volume}
		\left( k \cdot \frac{500n\log(M_+)}{M_+} \right)^{M_-} \leq \left( \frac{\rho}{1 + \rho} \right)^n |B_2^n|.
	\end{equation}

To set parameters, we note that the above is only successful if every point in $\bN$ satisfies $w_1 \cdot x < 0$, which happens with probability $k^{- M_{-}}$. Since we will need to repeat roughly $k^{M_-}$ times, we will need
		\[k^{M_-} \leq \poly(n,k,\rho^{-1}) 2^{O \left(\sqrt{n \log(1/\rho) \log(k)} \right)}\]
to bound the runtime. Balancing parameters, we can then set 
	\[ M_- := \sqrt{n \log(2/\rho) / \log(k)} \quad \quad M_+ = \poly(n,k, \log(\rho^{-1})) \cdot 2^{\sqrt{n \log(k) \log(2/\rho)}}.\]

\begin{figure}[t!]
    \centering
    \begin{subfigure}{0.49\textwidth}
		\includegraphics[width=\linewidth]{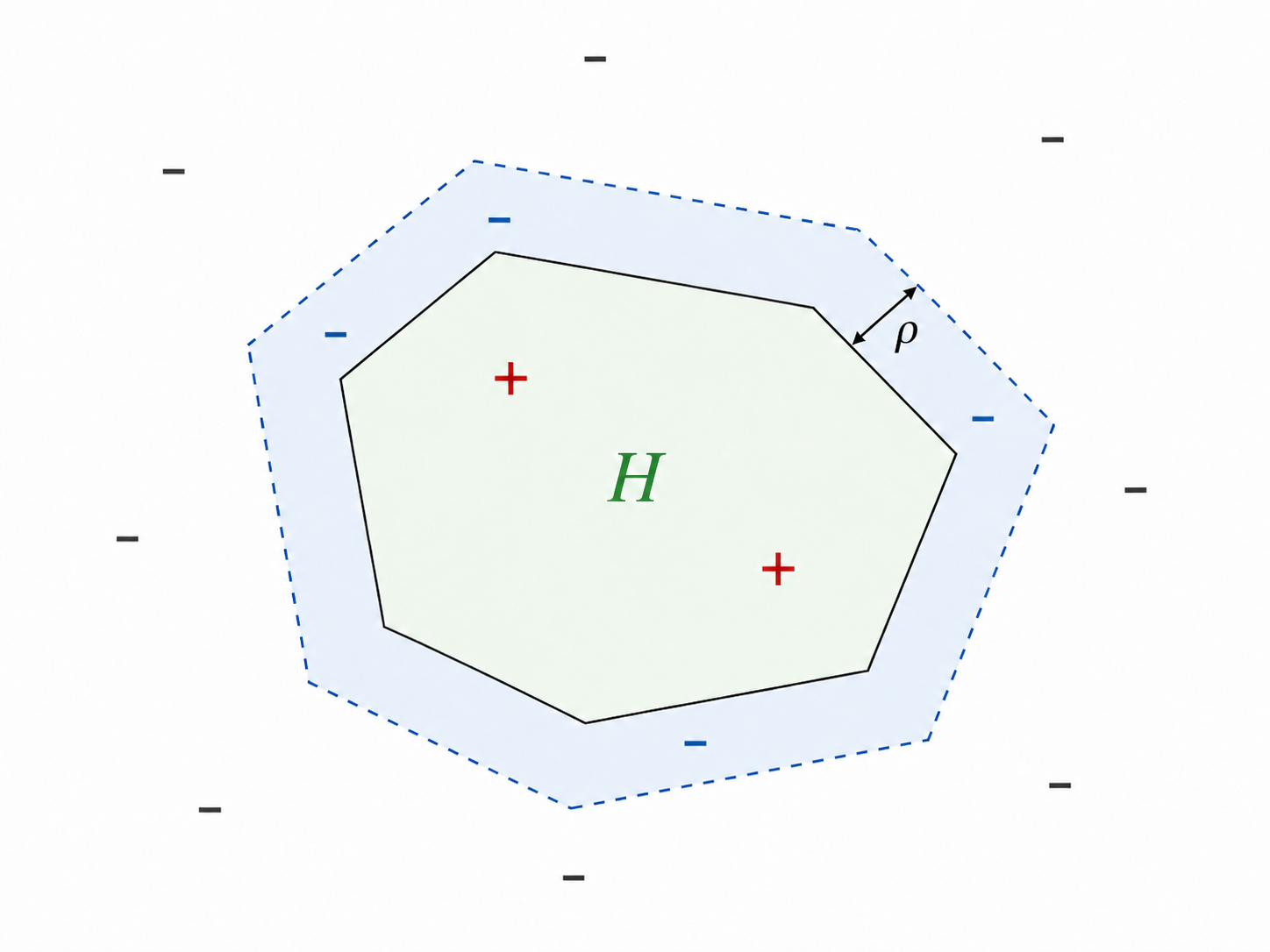}
        \caption{Soft Margin}
    \end{subfigure}
    \hfill
    \begin{subfigure}{0.49\textwidth}
		\includegraphics[width=\linewidth]{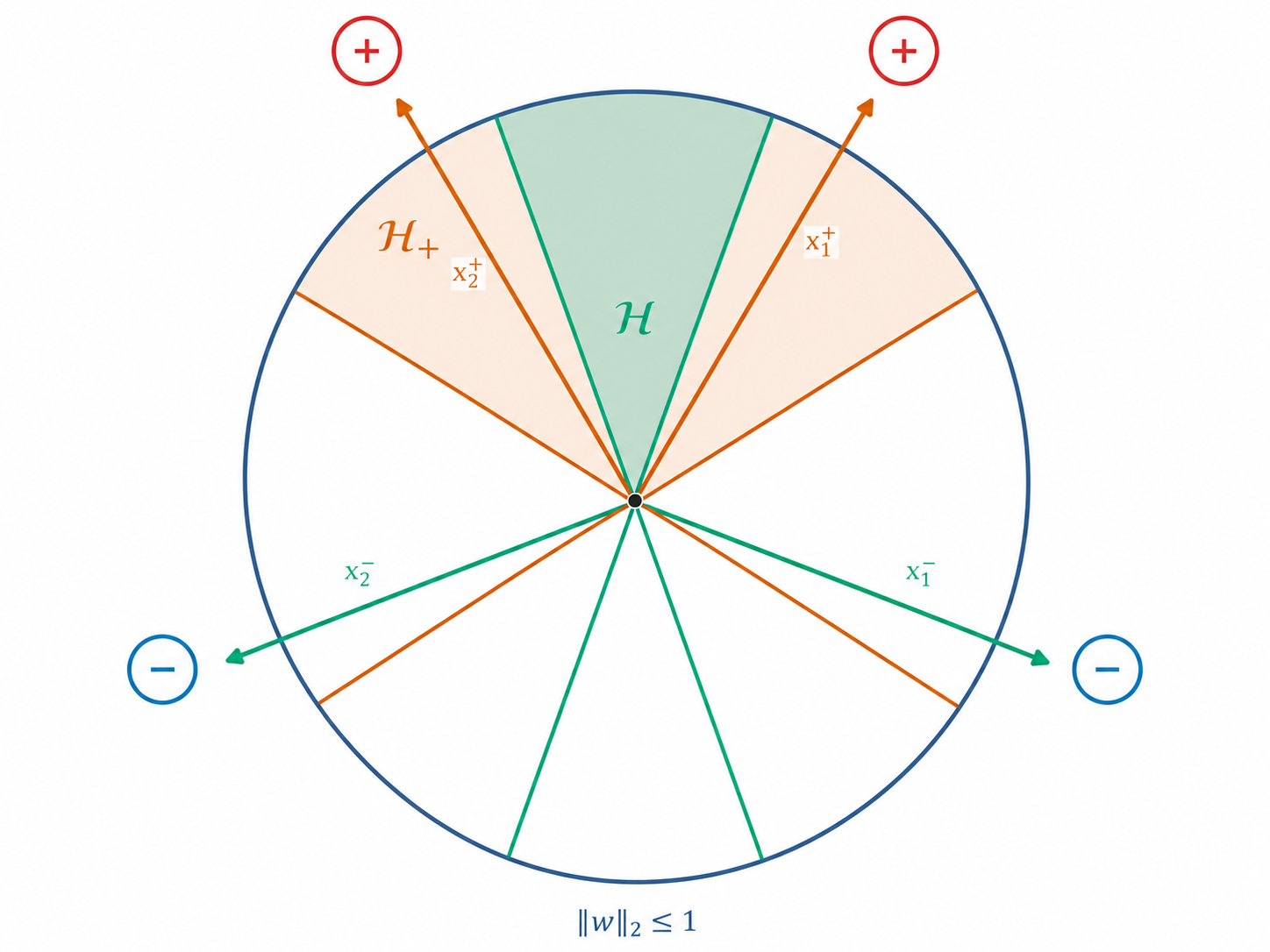}
        \caption{Good hypotheses}
    \end{subfigure}
    \hfill
    \caption{The first figure shows the soft margin in blue, i.e., points outside the intersection but within distance $\rho$. The second figure is the set of good hypothesis $\calH$ (shaded green) and the set $\calH_+$ of hypotheses that correctly classify the positive samples. Note that the $+$ and $-$ vectors correspond to positive and negatively labeled examples and lead to the orange and green halfspace constraints.}
    \label{fig:margin_and_algo}
\end{figure}

\subsubsection{Overview of \Cref{thm:main-full}}
To prove our soft margin results, we note that our previous learner only used $M_-=\sqrt{n \log(2/\rho)/\log(k)}$ negative samples to find a good halfspace $h$. Intuitively, so long as these negatively labelled samples all satisfy $w_1 \cdot x < -\rho$, the algorithm essentially should behave as though $f$ truly had a $\rho$ margin. While we must make some tweaks to the proof, this intuition can be made rigorous. 

One complication at this stage that arises is that our soft margin results will not be amenable to black box boosting arguments, particularly as we wish to learn using a hypothesis that is an intersection of halfspaces.\footnote{We note that if we didn't need to output an intersection of halfspaces, then we could apply the boosting argument from \cite{kalai2008agnostic}; we do not output an agnostic learner as in their theorem, but their proof appears to be robust to our setting of learning with ``soft margin''.}
To circumvent this, we simply note that by increasing the size of $M_+$ by a factor of $\poly(\eps^{-1})$ we can in fact find a halfspace $h$ satisfying 
\begin{align*}
&(i) \quad \Pr_{\bx \sim \calD}[h(\bx) = -1] \geq \poly(k,\rho^{-1}, \eps^{-1}) \cdot 2^{-O \left(\sqrt{n \log(k\eps^{-1}) \log(1/\rho)} \right)}	\\ \\
&(ii) \quad \Pr_{\bx \sim \calD}[f(\bx) = -1 |h(\bx) = -1] \geq 1 - \eps
\end{align*}

Crucially, $h$ has accuracy $1 - \eps$ on all points satisfying $h(x) = -1$. Thus, we can safely label these points negatively and recurse on the points satisfying $h(x) = 1$. In this way, we will build a hypothesis of the form $\bigwedge_{i=1}^t h_i(x)$. We continue recursing until either an $\eps$ fraction of the mass under $\calD$ satisfies $\bigwedge_{i=1}^t h_i(x) = 1$ or we are left with a distribution on which we cannot find a weak learner, i.e. a halfspace satisfying conditions $(i)$ and $(ii)$ above.

Intuitively, such a booster slowly covers the space with a union of regions of the form $\calR_i = \{x: h_i(x) = -1\}$. By Property $(ii)$, we have that a $1 - \eps$ fraction of points under $\calD$ in $\bigcup_i \calR_i$ satisfy $f(x) = -1$. This allows us to bound the false negative error rate of the final hypothesis, $\Pr_{\bx \sim \calD}[f(\bx) = 1 \land \bigwedge_{i=1}^T h_i(\bx) = -1]$, by $\eps$. To bound the false positive error, we note that we always succeed in producing a weak learner so long as an $\eps$ fraction of points satisfy $w_i \cdot x \leq -\rho$ for some $i \in [k]$. Thus, since our weak learner failed in the final iteration, all but an $\eps$ fraction of points $x$ satisfying $\bigwedge_{i=1}^T h_i(x) = 1$ and $f(x) = -1$ must have small margin. (In the case where there are only an $\eps$ fraction of points satisfying $\bigwedge_{i=1}^T h_i(x) = 1$, this is trivially true.) As the total fraction under $\calD$ of negatively labelled small margin points is at most $\nmarg(\rho,\calD,f)$, it follows that the false positive error rate, $\Pr_{\bx \sim \calD}\left [\bigwedge_{i=1}^T h_i(\bx) = 1 \land f(\bx) = -1 \right]$, is bounded by $\eta(\rho, \calD, f) + \eps$. Thus, this boosting procedure will output a hypothesis of error $\nmarg(\rho, \calD, f) + O(\eps)$, as desired.

\section{Preliminaries}

\textbf{Notation.} Throughout our discussion, we let $f(x) = \bigwedge_{i=1}^k \sign(w_i \cdot x)$ denote the target function, where we assume $\|w_i\|_2 = 1$ and that $\supp(\calD) \subseteq \mathbb{S}^{n-1}$ (as we'll see in \Cref{sec:margin}, these assumptions that the halfspaces are origin-centered, i.e., $\theta_i = 0$, and points lie on the unit sphere will be without loss of generality). Moreover, we let $\dpos$ and $\dneg$ denote $\calD$ conditioned on $f(\bx) = 1$ and $f(\bx) = -1$ respectively. For a convex body $\calK$, we let $|\calK|$ denote its Lebesgue volume. We denote random variables in boldface, e.g. $\br$.

\subsection{PAC Learning}
We now briefly recall the PAC learning model of Valiant \cite{valiant1984theory}. In this model, we wish to learn a class of Boolean-valued functions $\calC$ under an unknown distribution $\calD$. We say that an algorithm $(\eps, \delta)$-PAC learns the class $\calC$ if for any $f \in \mathcal{C}$ we have the following: given labeled examples of the form $(\bx, f(\bx))$ where $\bx \sim \mathcal{D}$, with probability $1 - \delta$, the algorithm outputs a hypothesis $h$ such that $\Pr_{\bx \sim \calD} [h(\bx) \neq f(\bx)] \leq \eps$. The complexity of a PAC-learner is measured by the number of samples it needs and the time it takes. 

We will use the following classical result.

\begin{lemma}[Hitting Set Lemma \cite{blumer1989learnability}]
\label{lem:hitting-set}
Let $\calD$ denote a distribution over $\calX$, $\calC$ denote a class of Boolean functions over $\calX$ of VC dimension at most $d$ and containing the constant $+1$ function, and $\eps \in (0,1/2)$, then with probability $1 - 2^{-d}$ a set $S$ of 
\[\frac{32}{\eps} \left( d \log \frac{1}{\eps} \right)\]
samples from $\calD$ satisfies $S \cap \{x \in \calX: c(x) = -1\} \not = \emptyset$ for all $c \in \calC$ such that $\Pr_{\bx \sim \calD} \left[ c(\bx) = -1 \right] \geq \eps$.
\end{lemma}

We will also need accuracy boosting for our results. In particular,

\begin{definition}[Weak Learner]
	We say that an algorithm $\calA$ is a weak learner with advantange $\gamma$ for a class $\calC$ if for any distribution $\calD$ and $f \in \calC$, with probability $3/4$, the algorithm outputs a hypothesis $h$ such that $\Pr_{\bx \sim \calD}[f(\bx) = h(\bx)] \geq \frac{1}{2} + \gamma$.
\end{definition}

Given such a hypothesis that does slightly better than random, we can get an algorithm with small error using accuracy boosting.

\begin{theorem}[\cite{schapire1990strength, freund1995boosting}]
\label{thm:boosting}
	Let $\eps,\delta \in (0,1)$ and suppose that $\calA$ is a weak learner for $\calC$ with advantage $\gamma$, then there is an algorithm that with probability $1 - \delta$ outputs a hypothesis $h$ such that $\Pr_{\bx \sim \calD}[f(\bx) \not = h(\bx)] \leq \eps$. Moreover, the algorithm runs in time $\poly(\eps^{-1}, \gamma^{-1}, \log(\delta^{-1}), T)$, where $T$ is the running time of the weak learner.
\end{theorem}

\subsection{Halfspaces}
\label{sec:margin}
The central object of study in this work will be halfspaces, i.e. functions of the form $\sgn(w \cdot x - \theta)$. We will also need the notion of a margin.  

\begin{definition}[Margin]
\label{def:margin}
	Given a distribution $\calD$ and an intersection of halfspaces $f(x) = \bigwedge_{i=1}^k \sgn(w_i \cdot x - \theta_i)$, where $\|w_i\|_2 = 1$ for all $i$, we say that $\calD$ has margin $\rho$ with respect to $f$ if for all $x \in \supp(\calD)$ with $f(x) = -1$ there exists some $i \in [k]$ such that
		\[\frac{w_i \cdot x - \theta_i}{R} \leq -\rho\]
	where $R = \sup_{x \in \supp(\calD)} \|x\|_2$. 
\end{definition}

Geometrically, such a notion implies that every negative point has distance at least $\rho \cdot R$ below some halfspace. This is the same notion used in \cite{gottlieb2021learning}, and differs slightly from the notion of robustness introduced by \cite{Arriaga2006algorithmic} and used in \cite{klivans2008margin}.

\begin{definition}[$\rho$-Robust Distribution]
\label{def:robust}
	We say that a distribution $\calD$ is $\rho$-robust with respect to a Boolean-valued function $f$ if 
		\[ \frac{\inf \{\|x-y\|_2: x \in \supp(\calD), y \in \mathbb{R}^n, f(x) \not = f(y)\}}{R} \geq \rho\]
	where $R = \sup_{x \in \supp(\calD)} \|x\|_2$.
\end{definition}

The difference is that our definition of margin is in terms of distance with respect to some supporting hyperplane while robustness measures distance to the polyhedron defined by $f$. Nevertheless, robustness implies a margin as in our definition. 

\begin{lemma}[Theorem 4.2 in \cite{gottlieb2021learning}]
\label{lem:robust-to-margin}
Let $f$ be an intersection of $k$ halfspaces and suppose that $\calD$ is a $\rho$-robust distribution with respect to $f$ and $\supp(\calD) \not \subset f^{-1}(-1)$, then $\calD$ has a $\frac{1}{2} \rho^2$ margin with respect to $f$.
\end{lemma}

We include a proof here for the reader's convenience.

\begin{proof}
Let $y \in \supp(\calD)$ be such that $f(y) = 1$. Note that since $\calD$ is robust, it follows that 
	\[\frac{w_i \cdot y - \theta_i}{R} \geq \rho\]
for all $i \in [k]$, as by robustness $f(y - \rho R w_i) = +1$. Now, to obtain a contradiction, suppose $\calD$ does not have a $\rho$ margin. It then follows that there exists a point $x \in \supp(\calD)$ such that
	\[\frac{w_i \cdot x - \theta_i}{R} \geq -\frac{1}{2}\rho^2\]
for all $i \in [k]$ and $f(x) = -1$. In this setting, we can define $z = \left(1- \frac{1}{2} \rho \right) x + \frac{1}{2} \rho \cdot y$. Clearly, $\|z-x\|_2 \leq \rho R$. We can also observe that for any $i \in [k]$
	\[w_i \cdot z - \theta_i = \left(1- \frac{1}{2} \rho \right) \cdot w_i \cdot x - \theta_i + \frac{1}{2} \rho \cdot (w_i \cdot y) \geq \left(1- \frac{1}{2} \rho \right) \cdot \left( w_i \cdot x - \theta_i \right) + \frac{1}{2} \rho^2 R \geq 0.\]
Thus, $f(z) = 1$, but this violates robustness. 
\end{proof}

The above notion of margin is relatively rigid and fails for most continuous distributions as a small amount of probability mass can be arbitrarily close to hyperplanes. Given this, we will also consider the following weaker notion of margin so that our results apply more generally. 

\begin{definition}[Soft Margin]\label{def:soft-margin}
	For $\rho \in (0,1)$,  the $\rho$-soft margin of distribution $\mathcal{D}$ with respect to an intersection of halfspaces $f = \bigwedge_{i=1}^k \sgn(w_i \cdot x - \theta_i)$ is
		\[\nmarg(\rho, \calD, f) := \Pr_{\bx \sim \calD} \left[ \min_{i \in [k]} \frac{w_i \cdot \bx - \theta_i}{R} \in [-\rho, 0] \right ] \]
	where $R = \sup_{x \in \supp(\calD)} \|x\|_2$.
\end{definition}

Intuitively, this quantity measures the fraction of points in the distribution that are ``barely" not in the polyhedron defined by the intersection of halfspaces, i.e., they are outside but within distance $\rho$ (such points would all have negative labels).\footnote{As a subtlety, this intuition holds up to a quadratic factor by \Cref{lem:robust-to-margin} so long as the polyhedron contains a ball of radius $\rho$. That said, for degenerate convex bodies the notion of distance to the body and soft margin could differ significantly.} This quantity also works for continuous distributions, and our algorithms will be able to learn arbitrary intersections, over arbitrary distributions, to accuracy $1 - \nmarg(\rho, \calD,f) - \eps$.

Finally, we will assume throughout that all points lie on the unit sphere and all halfspaces go through the origin. This can be achieved by first adding an $(n+1)$st coordinate equal to $\sqrt{R^2 - \|x\|_2^2}$  and rescaling so the vector has length $\frac{1}{\sqrt{2}}$. Afterwards, we add a $(n+2)$nd coordinate equal to $\frac{1}{\sqrt{2}}$ to each point $x$ and consider the weight vectors $w_i' = \frac{1}{\sqrt{1 + \theta_i^2/R^2}}(w_i, 0, -\theta_i/R)$. To see that this only changes the notions of margin and soft margin by at most a constant factor, note that we can assume that $|\theta_i| \leq R$ as otherwise $w_i \cdot x - \theta_i$ is constant. After this transformation, we can see that the margin and soft margin change by a factor of at most $2$ since
	\[|w_i' \cdot x| = \frac{\left| \frac{1}{\sqrt{2}\cdot R} w_i \cdot x - \frac{1}{\sqrt{2} \cdot R}\theta_i \right|}{\sqrt{1 + \theta_i^2/R^2}} \geq \frac{\left| w_i \cdot x - \theta_i \right|}{2 \cdot R}.\]

\subsection{Sampling from a Polytope}
We will crucially need to sample from a convex body in our algorithm. In particular, we use the following result for efficient sampling, which follows a long line of polynomial-time algorithms~\cite{dyer1991random,lovasz1993random,kannan1997random,lovasz2006simulated,lovasz2006hit,kannan2012random,cousins2018gaussian}. 
We note that since the specific polynomial complexity does not matter for our main result (whose complexity is sub-exponential in dimension), we could also use the classical rounding algorithm of \cite{grotschel2012geometric} based on the Ellipsoid method, followed by any polynomial-time uniform sampling algorithm such as the ball walk~\cite{kannan1997random} or hit-and-run~\cite{lovasz2006hit}.

\begin{theorem}[\cite{JLLV26reducing}]\label{thm:sampling}
Given a convex body $K \subset \R^n$ and numbers $r,R > 0$ s.t. $rB_2^n \subseteq K \subseteq RB_2^n$, where $B_2^n$ denotes the unit ball in $\R^n$ centered at the origin, for any $\varepsilon, \delta > 0$ there is a randomized algorithm that with probability at least $1-\delta$ outputs a point $x$ from a distribution that is within TV-distance $\varepsilon$ of the uniform distribution over $K$ using $\tilde{O}(n^{3.5}\log(R/r) \log(1/\varepsilon\delta))$ queries to the membership oracle and $\tilde{O}(n^2)$ arithmetic operations per query.    
\end{theorem}

The above is the complexity of the first random point. Subsequent points need only $\tilde{O}(n^2)$ queries per sample.  
We also note that the output guarantee has been strengthened from TV-distance to a pointwise guarantee, i.e., the density of the output distribution is at most $(1+\varepsilon)$ times the uniform density in the convex body~\cite{KV25sampling}.
\section{Learning Intersections of Halfspaces}

In this section, we will prove \Cref{thm:main-hard-margin} and \Cref{thm:main-full}. As discussed in the technical overview, our main tool to do so will be a procedure to find a halfspace $h$ such that the region below the halfspace, $h(x) = -1$, contains many more negative points than positive points. Applying an appropriate boosting argument then yields \Cref{thm:main-hard-margin} (standard boosting) and \Cref{thm:main-full} (boosting by covering).

\subsection{A Weak Learning Algorithm}

We begin by describing the \findh{} procedure, depicted in \Cref{alg:weak-learn} (illustrated in \Cref{fig:margin_and_algo}). 

{\small 
\begin{algorithm}[h]
\addtolength\linewidth{-2em}
\vspace{0.5em}

\textbf{Input:} Sample access to $\mathcal{D}$ with corresponding labels under $f(x) = \bigwedge_{i=1}^k \sign(w_i \cdot x)$, an accuracy parameter $\eps \in (0,1/2)$, a margin parameter $\rho \in (0,1/2)$ \\
\textbf{Output:} A hypothesis $h(x) = \sgn(w \cdot x)$. \\[0.5em]

\findh($\calD, f, \eps, \rho$):\\

\begin{enumerate}
	\item Set $M_{-} = \sqrt{n \log(9/\rho)/ \log(2k\eps^{-1})}$. 
Sample a set $\bN$ of $M_{-}$ examples from $\calD_{-}.$
	\item Set $M_{+} = \left(200 k \cdot \frac{n^2M_{-}}{\eps^4}\right)^2 \cdot 2^{\sqrt{n \log(9/\rho) \log(2k\eps^{-1})}}$.	
Sample a set $\bP$ of $M_+$ examples from $\dpos$. 
	\item Set $\calH$ to be the set of normal vectors of halfspaces that correctly classify both $\bP$ and $\bN$:
	\[\calH := \{w \in \mathbb{R}^{n+1}: \|w\|_2 \leq 1 \, \land \, w \cdot (x,1) \geq 0 \quad \forall x \in \bP \, \land \, w \cdot (x,1) \leq 0 \quad \forall x \in \bN\}.\]
	\item Sample $\bw$ uniformly at random from $\calH$ using \Cref{thm:sampling} (with total variation distance $0.01$, failure probability $0.01$, $r = 0.01\rho$, and $R = 1$)
	\item Return $h(x) = \sgn(\bw \cdot (x,1))$
\end{enumerate}
\caption{An Algorithm to Find a Good Halfspace}
\label{alg:weak-learn}
\end{algorithm}
}

The key property we will need from this algorithm is captured in the next lemma.
\begin{lemma}
\label{lem:weak-learn}
Suppose that $\min_{b \in \{\pm 1\}} \Pr_{\bx \sim \calD}[f(\bx) = b] \geq \eps/2$,
 \[\nmarg(\rho,\calD,f) \leq \Pr_{\bx \sim \calD}[f(\bx) = -1] - \eps/2,\]
 \sloppy then with probability at least $\Omega \left( 2^{-\sqrt{n \log(9/\rho) \log(2k\eps^{-1})}} \right)$ the hypothesis $h(x)$ output by $\findh(f,\calD, \eps, \rho)$ satisfies

$(i)$ $\Pr_{\bx \sim \calD} \left[ h(\bx) = -1 \right] \geq \exp \left( -2\sqrt{n \log(9/\rho) \log(2k\eps^{-1})} \right) \cdot \poly(\eps,k^{-1})$ 

$(ii)$ $\Pr_{\bx \sim \calD} \left[ f(\bx) = -1 \big |  h(\bx) = -1 \right] \geq 1 - \eps$
\end{lemma}

This immediately yields a weak learning algorithm as we know that $f(x)$ is biased towards $-1$ when $h(x) = -1$. Using a boosting algorithm then  gives a proof of \Cref{thm:main-hard-margin}. We note that proving \Cref{thm:main-full} will require more work, as we need a proper hypothesis and must be careful not to grow the set of small margin points.

\begin{proof}[Proof of \Cref{thm:main-hard-margin}]
If $f$ is noticeably biased under $\calD$, i.e. $\min_{b \in \{\pm 1\}} \Pr_{\bx \sim \calD}[f(\bx) = b] \leq 1/3$, then we can simply output a constant function to get a weak learner. Otherwise, by \Cref{lem:weak-learn}, $\findh(\calD, f, 0.01, \rho)$ outputs a hypothesis with advantage
		\[\exp \left( -6\sqrt{n \log(9/\rho) \log(k)} \right) \cdot \poly(k^{-1})\]
	with probability $p := 2^{-3\sqrt{n \log(9/\rho) \log(k)}}$. Thus, we can run \findh{} $np^{-1}$ times and check whether each hypothesis satisfies $(i)$ and $(ii)$ in \Cref{lem:weak-learn} to get a weak learner. We can then apply a boosting argument to finish the proof.
	
For the runtime bound, we simply note that \findh{} runs in time 
\[\exp \left(2\sqrt{n \log(9/\rho) \log(100k)}\right) \cdot \poly(k).\]
\Cref{thm:boosting} then yields the desired bound on the runtime after boosting.
\end{proof}

\subsection{Proof of \Cref{lem:weak-learn}}
We now turn to proving \Cref{lem:weak-learn}. To do so, it will be convenient to define 
\[\calH_+ := \{w \in \mathbb{R}^{n+1}: \|w\|_2 \leq 1 \land  w \cdot (x,1) \geq 0 \quad \forall x \in \bP\}.
\]
We can then note that
\[
\calH :=  \{w \in \calH_+: \|w\|_2 \leq 1 \land w \cdot (x,1) \leq 0 \quad \forall x \in \bN\}.
\]

We will now partition $\calH_+$ into ``good'' and ``bad'' regions. More generally, for any convex body $\calK \subseteq \mathbb{R}^{n+1}$, we denote
	\[\calK^{\good} = \left \{w \in \calK: \Pr_{\bx \sim \calD_-} [ w \cdot (\bx,1) < 0] \geq  \frac{100 n}{\eps^2}  \cdot \frac{\log(M_+)}{M_+}  \right\} \]
	and
	\[\calK^{\bad} := \calK \setminus \calK^{\good}.\]

The key fact will then be that any halfspace in $\hgood$ will satisfy the guarantees of \Cref{lem:weak-learn}.

\begin{lemma}
\label{lem:good-halfspaces}
Suppose that $\min_{b \in \{\pm 1\}} \Pr_{\bx \sim \calD}[f(\bx) = b] \geq \eps/2$ and
 \[\nmarg(\rho, \calD,f) \leq \Pr_{\bx \sim \calD}[f(\bx) = -1] - \eps/2,\]
 then with high probability over $\bP$ we have that any $w \in \hgood$ satisfies properties $(i)$ and $(ii)$ of \Cref{lem:weak-learn}. 
\end{lemma}
  
\begin{proof}
We start by observing that by definition any halfspace $w \in \hgood$ satisfies
	\[\Pr_{\bx \sim \dneg} \left[ \sgn(w \cdot (\bx,1))< 0 \right] \geq \frac{100n}{\eps^2} M_+^{-1} \log(M_+) \geq 2^{-2\sqrt{n \log(9/\rho) \log(2k\eps^{-1})}} \cdot \poly(\eps, k^{-1})\]
by our choice of $M_+$, yielding property $(i)$.

To prove property $(ii)$, we note that with high probability over $\bP$ by \Cref{lem:hitting-set} for all $w \in \hgood$
\[\Pr_{\bx \sim \dpos} \left[ w \cdot (\bx,1)  < 0 \right] \leq \frac{50 n \log(M_+)}{M_+} \]
where we used the fact that halfspaces over $\R^{n+1}$ have VC dimension $n+2$ and 
\[ \frac{32 M_+}{50 n \log(M_+)} \left( (n+2) \log \frac{M_+}{50n \log(M_+)} \right) \leq M_+.\]

We can then compute that for $h(x) = \sgn(w \cdot (x,1))$ with $w \in \hgood$
\begin{align*}
	\Pr_{\bx \sim \calD} \left[f(\bx) = 1 | h(\bx) = -1 \right] &= \Pr_{\bx \sim \calD} \left[ h(\bx) = -1 | f(\bx) = 1 \right] \cdot \frac{\Pr_{\bx \sim \calD}[f(\bx) = 1]}{\Pr_{\bx \sim \calD}[h(\bx) = -1]} \\
	&\leq \Pr_{\bx \sim \calD} \left[ h(\bx) = -1 | f(\bx) = 1 \right] \cdot \frac{\Pr_{\bx \sim \calD}[f(\bx) = 1]}{(100 n/\eps^2 \cdot M_+^{-1} \log(M_+)) \Pr_{\bx \sim \calD}[f(\bx) = -1]} \\
	&\leq \frac{M_+ \cdot \eps}{50 n \log(M_+)} \Pr_{\bx \sim \dpos} \left[ h(\bx) = -1 \right] \\
	&\leq \eps 
\end{align*}
where the first inequality uses the definition of $\hgood$ and the second inequality follows from the fact that $\min_{b \in \{\pm 1\}} \Pr_{\bx \sim \calD}[f(\bx) = b] \geq \eps/2$.
\end{proof}

Now that any $w \in \hgood$ suffices, our goal will be to show that the output satisfies $\bw \in \hgood$ with non-trivial probability. Towards this, we will show that mandating that the points in $\bN$ are labelled negatively will shrink the volume of the bad hypotheses significantly. On the other hand, we can note that volume of $\calH$ is relatively large when all the samples in $\bN$ are negatively labelled by some halfspace in the target function. Note that such an assumption on $\bN$ is vital as otherwise it is not even clear a priori that $\calH$ is nonempty.

\begin{lemma}
\label{lem:vol-lb}
Let $f(x) = \bigwedge_{i=1}^k \sign(w_i \cdot x)$, if $w_1 \cdot x < -\rho$ for all $x \in \bN$, then
	\[|\calH| \geq \left(\frac{\rho}{3(1+\rho)} \right)^{n+1} |B_2^n|\]
where $B_2^n = \{w: \|w\|_2 \leq 1\}$ denotes the unit ball in $\R^n$. Moreover, $\calH$ contains a ball of radius $\frac{\rho}{6(1 + \rho)}$.
\end{lemma}

\begin{proof}
We claim that
	\[\left( w_1 + \frac{\rho}{3} B_2^n \right) \times \left [\frac{\rho}{3} , \frac{2\rho}{3} \right] \subseteq (1 + \rho) \cdot \calH \]
Since the volume of this cylinder is $\left(\frac{\rho}{3} \right)^{n+1} |B_2^n|$ and it contains a ball of radius $\frac{\rho}{6}$, the lemma will immediately follow. To prove the claim, note that if $w \in \left( w_1 + \frac{\rho}{3} B_2^n \right) \times \left [\frac{\rho}{3} , \frac{2\rho}{3} \right]$, then by the triangle inequality
	\[\|w\| \leq \|w_1\| + \frac{\rho}{3} + \frac{2\rho}{3} = 1 + \rho\]
We'll now let $w = (w', \theta)$ where $w' \in \R^n$ denotes the first $n$ coordinates of $w$ and $\theta \in \R$ is the $(n+1)$st coordinate. We can then compute that for any $x \in \bP$,
	\[w \cdot (x,1) = w' \cdot x + \theta \geq w' \cdot x + \frac{\rho}{3} \geq w_1 \cdot x - \frac{\rho}{3} + \frac{\rho}{3} > 0.\]
Now as $w_1 \cdot x < -\rho$ for all $x \in \bN$, we have that if $x \in \bN$ then
	\[w \cdot (x,1) = w' \cdot x + \theta \leq w' \cdot x + \frac{2 \rho}{3} \leq w_1 \cdot x + \frac{\rho}{3} + \frac{2\rho}{3} < 0.\]
Thus, $w \in (1 + \rho) \cdot \calH$ as desired.
\end{proof}

We now turn to show that the volume of the bad points in $\calH_+$ will be small, via a simple averaging argument. 

\begin{lemma}
\label{lem:vol-decay}
Let $\calK \subseteq \mathbb{R}^{n+1}$ be a bounded, convex body and for $x \in \R^n$ let
	\[\calK_x = \calK \cap \{ w: w \cdot (x,1) < 0 \}.\]
If $\calD^{(1)}_{-}$ denotes $\dneg$ conditioned on $w_1 \cdot \bx < -\rho$, then
	\[\Pr_{\bx \sim \dneg}[w_1 \cdot \bx < -\rho] \cdot \E_{\bx \sim \calD^{(1)}_{-}} \left[ \left|\calK_{\bx}^{\bad} \right| \right] \leq \frac{100 n}{\eps^2} \cdot M_+^{-1} \log(M_+) \cdot \left|\calK^{\bad} \right| \]
\end{lemma}

\begin{proof}
	Let $p(x)$ denote the pdf of $\dneg$ and observe
		\begin{align*}
		\Pr_{\bx \sim \dneg}[w_1 \cdot \bx < -\rho] \cdot \E_{\bx \sim \dneg^{(1)}} \left[ \left| \calK_{\bx}^\bad \right| \right] &\leq \E_{\bx \sim \dneg} \left[ \left| \calK_{\bx}^\bad \right| \right] \\
			 &= \int_{\mathbb{R}^n} p(x) \int_{\calK_{\bx}^{\bad}} 1 \;\; dw \; dx \\ 
				&= \int_{\mathbb{R}^n} p(x) \int_{\calK^{\bad}} \mathbb{I}[w \cdot (x,1) < 0] \;\; dw \; dx \\
				&= \int_{\calK^{\bad}} \int_{\mathbb{R}^n} p(x) \cdot \mathbb{I}[w \cdot (x,1) < 0] \;\; dx \; dw \\
				&= \int_{\calK^{\bad}} \Pr_{\bx \sim \dneg}[w \cdot (\bx,1) < 0] \;\; dw \\
				&\leq \frac{100 n}{\eps^2} \cdot M_+^{-1} \log(M_+) \cdot \left| \calK^{\bad} \right|
		\end{align*}
		
\end{proof}

We can now prove the main lemma.

\begin{proof}[Proof of \Cref{lem:weak-learn}]
	Let $f(x) = \bigwedge_{i=1}^k \sign(w_i \cdot x)$, $\nmarg = \nmarg(\rho,\calD,f)$ and without loss of generality suppose
		\[\Pr_{\bx \sim \dneg} \left[ w_1 \cdot \bx < -\rho\right] \geq \frac{1}{k} \cdot \left( 1 - \frac{\nmarg}{\Pr_{\bx \sim \calD}[f(\bx) = -1]} \right). \]
	We now condition on the event that all samples from $\bN$ are drawn from $\dheavy$, where $\dheavy$ is $\dneg$ conditioned on $w_1 \cdot \bx < -\rho$. Note that this happens with probability at least
	\begin{align*}
		\left( \frac{1}{k} \cdot \left( 1 - \frac{\nmarg}{\Pr_{\bx \sim \calD}[f(\bx) = -1]} \right) \right)^{M_{-}} &\geq \left( \frac{1}{k} \cdot \left( 1 - \frac{\nmarg}{\nmarg + \eps/2} \right) \right)^{M_{-}} \\
		&\geq \left( \frac{\eps}{2k} \right)^{M_-} \\
		&= 2^{-\sqrt{n \log(9\rho^{-1}) \log(2k\eps^{-1})}}
	\end{align*}

By \Cref{lem:vol-lb}, we now have that
			\[|\calH | \geq \left(\frac{\rho}{3(1+\rho)} \right)^{n+1} |B_2^n| \geq \left(\frac{\rho}{6} \right)^{n+1} |B_2^n|. \]	
	On the other hand, \Cref{lem:vol-decay} yields that for any convex body
	\begin{align*}
		\E_{\bx \sim \dheavy} [ |\calK_{\bx}^\bad| ] &\leq \left(\Pr_{\bx \sim \dneg} \left[ w_1 \cdot \bx < -\rho \right]\right)^{-1} \cdot \frac{100 n}{\eps^2} \cdot M_+^{-1} \log(M_+) \cdot |\calK^\bad| \\
		&= \left(\Pr_{\bx \sim \dneg} \left[ w_1 \cdot \bx < -\rho \right]\right)^{-1} \cdot \frac{100 n}{\eps^2} \log(M_+) \cdot \left(\frac{1}{200 k} \cdot \frac{\eps^4}{n^2M_{-}} \right)^2 \cdot 2^{-\sqrt{n \log(9/\rho) \log(2k\eps^{-1})}} \cdot |\calK^\bad| \\
		&\leq \frac{1}{nM_{-}} \cdot 2^{-\sqrt{n \log(9/\rho) \log(2k\eps^{-1})}} \cdot |\calK^\bad|
	\end{align*}
	Thus, by Markov's inequality, we have that
		\[\Pr_{\bx \sim \dheavy} \left[ |\calK_{\bx}^\bad| \geq 
2^{-\sqrt{n \log(9/\rho) \log(2k\eps^{-1})}} \cdot |\calK^\bad| \right] \leq \frac{1}{n M_-} \]
	So with high probability
		\begin{align*}
|\calH^\bad| &\leq \left( 
2^{-\sqrt{n \log(9/\rho) \log(2k\eps^{-1})}} \right)^{M_{-}} |\hgood| \\ 
&\leq \left( 2^{-\sqrt{n \log(9/\rho) \log(2k\eps^{-1})}} \right)^{M_{-}} \cdot |B_2^n| \ll \left ( \frac{\rho}{6} \right)^n |B_2^n| 
\end{align*}	
	Thus, conditioned on drawing $\bN$ from $\dheavy$, it follows that a random vector from $\calH$ satisfies $w \in \hgood$ with high probability.
	
	To complete the proof, note that since $\calH$ contains a ball of radius $\rho/12$ by \Cref{lem:vol-lb}, it follows by \Cref{thm:sampling}, that we draw a $w \in \hgood$ with probability at least $0.95$. Thus, by \Cref{lem:good-halfspaces}, $h$ satisfies properties $(i)$ and $(ii)$ with the desired success probability.
\end{proof}

\subsection{Boosting by Covering}

We now turn to describing how to get a strong learner from \findh{} when we only have a soft margin. To do so, it will be convenient to work with the following abstraction for simplicity.

\begin{definition}[Region Learner]
\label{def:region}
	Let $\calC$ denote a family of boolean functions and $\calF_f$ denote a family of distributions parameterized by $f \in \mathcal{C}$. We say that an algorithm is an $\eps$-region learner for $\calC$ with advantage $\gamma$ if for any $f \in \calC$ and any distribution $\calD \in \calF_f$, we have that the algorithm, with probability at least $0.99$, outputs a region $\calR$ such that

	$(i)$ $\Pr_{\bx \sim \calD}[\bx \not \in \calR] \geq \gamma$
	
	$(ii)$ $\Pr_{\bx \sim \calD}[f(\bx) = -1 |\bx \not \in \calR] \geq 1 - \eps$
	
\noindent using labeled examples from $\calD$ and sample access to $\dneg$ and $\dpos$.
\end{definition}

By \Cref{lem:weak-learn}, we can get a region learner for intersections of halfspaces by running \findh{} many times and checking if each output satisfies $(i)$ and $(ii)$. In this case, our set $\calF_f$ will correspond to the distributions that satisfy the assumptions of \Cref{lem:weak-learn}.

We now show how to boost a region learner to get a strong learner. In particular, we will simply repeatedly find a region $\calR$ using our region learner and then restrict our attention to points in $\calR$, as we know we can safely label points outside of $\calR$ negatively. That said, some additional care is needed to ensure that we can sample from the restricted distributions efficiently.

{\small
\begin{algorithm}[h]
\addtolength\linewidth{-2em}
\vspace{0.5em}

\textbf{Input:} Sample access to $\mathcal{D}$ with corresponding labels under $f \in \calC$, an accuracy parameter $\eps \in (0,1/2)$, $\calA$ an $\eps$-region learner with advantage $\gamma$ \\
\textbf{Output:} A hypothesis $h(x) = (x \in \calR_1) \land \dots \land (x \in \calR_t)$ \\[0.5em]

\cover($\calD, f, \calA, \eps, \gamma$):\\

\begin{enumerate}
	\item \label{line:balance} If $\Pr_{\bx \sim \calD}[f(\bx) = b] \leq 5\eps$ for some $b \in \bits$, return the constant function $h \equiv -b$
	\item $t \gets 0$, $\calD_0 \gets \calD, g_0 \equiv 1$
	\item \label{line:loop} While $t < \gamma^{-1} \log(1/\eps)$ and $\min_{b \in \bits} \Pr_{\bx \sim \calD} \left[f(\bx) = b \land \bigwedge_{i=0}^t \bg_i(\bx) = 1 \right] \geq \eps$:
	\begin{enumerate}
		\item \label{line:compute-region} Run $\calA$ on $\calD_{t}$ and $f$ and let $\calR_{t+1}$ be the region it outputs and $g_{t+1}(x)$ be the indicator of $x \in \calR_{t+1}$
		\item \label{line:ret-bad} Check if $\calR_{t+1}$ satisfies $(i)$ and $(ii)$ in \Cref{def:region} by using a fresh set of samples. If not, return $h(x) = \bigwedge_{i=0}^{t} g_i(x)$ 
		\item Set $\calD_{t+1}$ to be $\calD_{t}$ conditioned on $g_{t+1}(\bx) = 1$
		\item $t \gets t+1$
		\end{enumerate}
	\item \label{line:ret-good} Output $h(x) := \bigwedge_{i = 0}^{t} g_i(x)$
\end{enumerate}
\caption{An Algorithm to Boost a Region Learner}
\label{alg:cover-learn}
\end{algorithm}
}

Our main goal will then be to show the following guarantee. 

\begin{lemma}
\label{lem:cover-learn}
Let $\calA$ be an $\eps$-region learner with advantage $\gamma$ on a set of distributions $\calF_f$. Then, with probability at least $0.99$, $\cover(\calD, f, \calA, \eps, \gamma)$ outputs a hypothesis with error at most
	\[\begin{cases}  \Pr_{\bx \sim \calD} \left[f(\bx) = -1 \land h(\bx) = +1 \right] + \eps & \text{ $h(x)$ is returned on Line \ref{line:ret-bad}} \\ 6 \eps & \text{ otherwise} \end{cases}.\]
Moreover, if we are in the first case, we have that
	\[\min_{b \in \bits} \Pr_{\bx \sim \calD} \left[f(\bx) = b \land h(\bx) = 1 \right] \geq 2\eps/3\]
\end{lemma}

Before moving to the prove the statement, we remark that the guarantee is somewhat unusual compared to standard boosting results. In particular, the error of the learner depends on whether we return on Line \ref{line:ret-bad}. Notably, this is necessary. Unlike in the standard boosting setting, our region learner may suddenly face a distribution that is not in $\calF_f$. In this scenario, we then have no guarantee on the output of the region learner and must terminate.

At a high-level, the proof will follow by bounding the false negative error rate, $\Pr_{\bx \sim \calD}[h(\bx) = -1 \land f(\bx) = 1]$, using the fact that $f$ is highly biased towards $-1$ on $\calR_i \setminus \calR_{i+1}$ (property $(ii)$ of \Cref{def:region}). We will then bound the false positive error rate, $\Pr_{\bx \sim \calD}[h(\bx) = 1 \land f(\bx) = -1]$, by taking cases on where $h(x)$ is returned in \cover{}. Towards this, we first show that the algorithm never exits the loop on Line \ref{line:loop} because of $\Pr_{\bx \sim \calD} \left[f(\bx) = 1 \land \bigwedge_{i=0}^t g_i(\bx) = 1 \right]$ is too small. 

\begin{lemma}
\label{lem:pos-mass-big}
	Suppose that the check on Line \ref{line:ret-bad} is always correct and $\min_{b \in \bits} \Pr_{\bx \sim \calD} [f(\bx) = b] \geq 4\eps$, then for all times $t$ we have
		\[\Pr_{\bx \sim \calD} \left[f(\bx) = 1 \land \bigwedge_{i=0}^t g_i(\bx) = 1 \right] \geq 2 \eps\]
\end{lemma}

\begin{proof}
Denote $h_k = \bigwedge_{i=0}^k g_i(x)$ and define
	\[\delta_k^+ := \Pr_{\bx \sim \calD}[f(\bx) = 1 \land h_{k+1}(\bx) = 1] - \Pr_{\bx \sim \calD}[f(\bx) = 1 \land h_k(\bx) = 1] \]
and
	\[\delta_k^- := \Pr_{\bx \sim \calD}[f(\bx) = -1 \land h_{k+1}(\bx) = 1] - \Pr_{\bx \sim \calD}[f(\bx) = -1 \land h_k(\bx) = 1]. \]
Now observe that
	\[\delta_k^+ = -\Pr_{\bx \sim \calD}[f(\bx) = 1 \land g_{k+1}(\bx) = -1 \land h_k(\bx) = 1]\]
and 
	\[\delta_k^- = -\Pr_{\bx \sim \calD}[f(\bx) = -1 \land g_{k+1}(\bx) = -1 \land h_k(\bx) = 1].\]
By property $(ii)$ of \Cref{def:region}, we have that
\begin{align*}
	\Pr_{\bx \sim \calD_{k}} \left[ f(\bx) = -1 | \bx \not \in \calR_{k+1} \right] &\geq \frac{1 - \eps}{\eps} \Pr_{\bx \sim \calD_{k}} \left[ f(\bx) = +1 | \bx \not \in \calR_{k+1} \right] \\
	&\geq \frac{1}{2\eps} \Pr_{\bx \sim \calD_{k}} \left[ f(\bx) = +1 | \bx \not \in \calR_{k+1} \right]
\end{align*}
where we used that $\eps \leq \frac{1}{2}$. We then conclude
	\[\delta_k^- \leq \frac{1}{2\eps} \cdot \delta_k^+.\]

We can now complete the proof by noting that for any $t$
	\[-1 \leq \sum_{k = 0}^{t-1} \delta_k^- \leq \frac{1}{2\eps} \cdot \sum_{k=0}^{t-1} \delta_k^+ \]
which implies that
	\[\Pr_{\bx \sim \calD}[f(\bx) = 1 \land h_{t}(\bx) = 1] - \Pr_{\bx \sim \calD}[f(\bx) = 1 \land h_0(\bx) = 1] \geq -2\eps.\]
Using the premise that $\Pr_{\bx \sim \calD} [f(\bx) = +1] \geq 4\eps$ then finishes the proof.
\end{proof}

With this we can now prove \Cref{lem:cover-learn}.

\begin{proof}[Proof of \Cref{lem:cover-learn}]
	We will condition on the event that the check on Line \ref{line:ret-bad} is always correct, the check on Line \ref{line:balance} estimates $\min_{b \in \bits} \Pr_{\bx \sim \calD} [f(\bx) = b]$ to accuracy $\eps$, and the check on Line \ref{line:loop} estimates $\min_{b \in \bits} \Pr_{\bx \sim \calD} \left[f(\bx) = b \land \bigwedge_{i=0}^t \bg_i(x) = 1 \right]$ to accuracy $\eps/3$. Notably, assuming we draw $\poly(\eps^{-1}, \gamma^{-1})$ samples for these checks, this happens with probability at least $0.99$ over the at most $O(\gamma^{-1} \log(1/\eps))$ such checks.
	We now note that if \cover{} returns on Line \ref{line:balance}, then we clearly have advantage at least $6\eps$. As such, we'll assume Line \ref{line:balance} does not return and thus
		\[\min_{b \in \bits} \Pr_{\bx \sim \calD} [f(\bx) = b] \geq 4\eps.\]

	Letting $T$ denote the final time $t$ in \cover{} and $h_k = \bigwedge_{i=0}^k g_i(x)$, we then note
	\begin{align*}
		\Pr_{\bx \sim \calD}[f(\bx) = +1 \land h(\bx) = -1] &= \sum_{i=1}^T \Pr_{\bx \sim \calD}[f(\bx) = +1 \land g_i(\bx) = -1 \land \forall j < i ~~ g_j(\bx) = 1] \\
		&= \sum_{i=1}^T \Pr_{\bx \sim \calD}[ g_i(\bx) = -1 \land h_{i-1}(\bx) = 1] \cdot \Pr_{\bx \sim \calD_{i-1}}[f(\bx) = +1 | g_i(\bx) = -1 ] \\
		&\leq \sum_{i=1}^T \eps \Pr_{\bx \sim \calD}[ g_i(\bx) = -1 \land h_{i-1}(\bx) = 1 ] \\
		&= \eps \Pr[h(\bx) = -1] \\
		&\leq \eps
	\end{align*}
where the first inequality used property $(ii)$ in \Cref{def:region}. 

This immediately implies that the error $h(x)$ is bounded by 

\[\Pr_{\bx \sim \calD} [f(\bx) = -1 \land h(\bx) = +1] + \eps\]

\noindent when $h(x)$ is returned on Line \ref{line:ret-bad}.

Thus, it remains to show

	\begin{equation}
	\label{eq:fp-error}
		\Pr_{\bx \sim \calD} \left[f(\bx) = -1 \land h(\bx) = +1 \right] \leq \eps.
	\end{equation}
when $h(x)$ is returned on Line \ref{line:ret-good}, as adding this to the false negative error that we just computed would complete the proof of the lemma. To bound \Cref{eq:fp-error}, first suppose that we exit the loop on Line \ref{line:loop} because 
	\[\min_{b \in \bits} \Pr_{\bx \sim \calD} \left[f(\bx) = b \land \bigwedge_{i=0}^T \bg_i(\bx) = 1 \right] \leq \eps.\]
By \Cref{lem:pos-mass-big}, we must have that $b = -1$ and thus \Cref{eq:fp-error} holds as desired.

Now suppose that we exit the loop because $T \geq \gamma^{-1} \log(1/\eps)$. We then note
\begin{align*}
	\Pr_{\bx \sim \calD}[h_k(\bx) = 1] &= \Pr_{\bx \sim \calD}[h_{k-1}(\bx) = 1 \land g_k(\bx) = 1] \\
	&= \Pr_{\bx \sim \calD}[h_{k-1}(\bx) = 1] \Pr_{\bx \sim \calD_{k-1}}[g_k(\bx) = 1] \\
	&\leq (1 - \gamma) \Pr_{\bx \sim \calD}[h_{k-1}(\bx) = 1]
\end{align*}
by property $(i)$ of \Cref{def:region}.
Thus, 
	\[\Pr_{\bx \sim \calD} \left[f(\bx) = -1 \land h(\bx) = +1 \right] \leq (1 - \gamma)^T \Pr_{\bx \sim \calD} [h(\bx) = 1] \leq (1 - \gamma)^T \leq e^{-\gamma T} \leq \eps.\]

\end{proof}

We conclude by recording the following lemma capturing the runtime and sample complexity aspects of \cover{}.

\begin{lemma}
\label{lem:cover-runtime}
$\cover(\calD, f, \calA, \eps, \gamma)$ calls $\calA$ at most $O(\gamma^{-1} \log(\eps^{-1}))$ times and with probability at least $0.98$ uses at most $\poly(\eps^{-1}, \gamma^{-1}) m_\calA$ samples with high probability, where $m_\calA$ is the sample complexity of $\calA$. Additionally, the hypothesis will be of the form $(x \in \calR_1) \land \dots \land (x \in \calR_t)$, where each $\calR_i$ is either produced by the region learner, the empty set, or the entire domain.
\end{lemma}

\begin{proof}
Note that the form of the hypothesis is immediate by inspection of \cover{}. The bound on the number of calls to $\calA$ follows from the bound on $t$ in Line \ref{line:loop}. To bound the number of samples, note that the checks on Lines \ref{line:balance}, \ref{line:loop} use at most $\poly(\eps^{-1}, \gamma^{-1})$ samples from $\calD$. Assuming our estimate of
	\[\min_{b \in \bits} \Pr_{\bx \sim \calD} \left[f(\bx) = b \land \bigwedge_{i=0}^t \bg_i(\bx) = 1 \right]\]
on in Line \ref{line:loop} is always correct to accuracy $\eps/2$, which happens with probability at least $0.99$, it follows that samples from $\calD_{t}, (\calD_{t})_{-}, (\calD_{t})_{+}$ can be drawn using $O(\eps^{-1})$ samples from $\mathcal{D}$ in expectation. Thus, with probability $0.99$, the total number of samples used in Line \ref{line:compute-region} and in the check on Line \ref{line:ret-bad} over all iterations of the loop is at most $\poly(\eps^{-1}, \gamma^{-1}) \cdot m_\calA$ by Markov's inequality.
\end{proof}

\subsection{Proof of \Cref{thm:main-full}}

With this, we now have everything that we need to prove \Cref{thm:main-full}.

\begin{proof}[Proof of \Cref{thm:main-full}]
	We consider running $\findh$ with accuracy $0.1\eps$. We then let $p := 2^{-O \left(\sqrt{n \log(1/\rho) \log(k\eps^{-1})} \right)}$ denote the success probability in \Cref{lem:weak-learn} and 
	\[\gamma := \exp \left( -O \left(\sqrt{n \log(1/\rho) \log(k\eps^{-1})} \right) \right) \cdot \poly(\eps,k^{-1})\]
	 be the quantity in the lower bound from $(i)$ in \Cref{lem:weak-learn}.
	
	It now follows that running \findh{} $\frac{kn}{\rho \eps} \cdot p^{-1}$ times, drawing $\frac{kn}{\rho \eps \gamma^2}$ fresh samples for each iteration, and verifying whether each hypothesis output by $\findh{}$ satisfies $(i)$ and $(ii)$ constitutes an $0.1\eps$ region learner, where the regions output are halfspaces. Moreover, on any distribution satisfying the premise of \Cref{lem:weak-learn}, we have that the algorithm fails to output a halfspace satisfying $(i)$ and $(ii)$ in \Cref{def:region} with probability $\exp(-\Omega(kn/(\eps \rho))) \ll \gamma^{-1} \log(\eps^{-1})$.
	
	We now run \cover{} with error parameter $0.1\eps$ to produce a hypothesis $h(x)$. By a union bound on the probability that $\findh{}$ fails to output a good region and \Cref{lem:cover-learn}, it follows that with probability $0.98$ 
		\[\Pr_{\bx \sim \calD} [h(\bx) \not = f(\bx)] \leq 	\begin{cases}  \Pr_{\bx \sim \calD} \left[f(\bx) = -1 \land h(\bx) = +1 \right] + 0.1\eps & \text{ $h(x)$ is returned on Line \ref{line:ret-bad}} \\ 0.6 \eps & \text{ otherwise} \end{cases}. \]
	and in the first case, where \cover{} returns on Line \ref{line:ret-bad}, we have that $\calD_{T}$ violates the assumptions of \Cref{lem:weak-learn} and satisfies
	\[\min_{b \in \bits} \Pr_{\bx \sim \calD} \left[f(\bx) = b \land h(\bx) = 1 \right] \geq \frac{2}{3} \cdot (0.1 \eps) \]
	where $T$ denotes the final value of $T$ in \cover{}.
	Note that this second condition implies that for any $b \in \bits$ 
		\[\Pr_{\bx \sim \calD_{T}} \left[f(\bx) = b \right] = \Pr_{\bx \sim \calD} \left[f(\bx) = b | h(\bx) = 1 \right] \geq \frac{2}{3} \cdot (0.1 \eps).\]
	Thus, if the premise of \Cref{lem:weak-learn} is violated, we must have that
	 \[\nmarg(\rho,\calD_{T},f) \geq \Pr_{\bx \sim \calD_{T}}[f(\bx) = -1] - 0.05\eps,\]
	Rearranging then yields
	 	\[\Pr_{\bx \sim \calD_T}[f(\bx) = -1] \leq \nmarg(\rho,\calD_T,f) + 0.05 \eps.\]
	 which implies
	 \begin{align*}
	 	\Pr_{\bx \sim \calD}[f(\bx) = -1 \land h(\bx)=1] &\leq \nmarg(\rho,\calD_T,f) \Pr_{\bx \sim \calD} [ h(\bx) = 1] + 0.05 \eps \\
	 	&= \Pr_{\bx \sim \calD} \left[h(\bx) = 1 \land \min_{i \in [k]} w_i \cdot \bx \in [-\rho, 0] \right] + 0.05\eps \\
	 	&\leq \eta(\rho, \calD, f) + 0.05\eps
	 \end{align*}
	 Thus, we have that with probability $0.98$, $h(x)$ has error at most $\nmarg(\rho,\calD,f) + \eps$, as desired. 
	 
	For the runtime, we note that by taking another union bound, by \Cref{lem:cover-runtime}, we have that with probability $0.96$, we draw at most
	 \[\poly(\eps^{-1}, \gamma^{-1}, k, \rho^{-1}) \cdot 2^{O\left(\sqrt{n \log(1/\rho) \log(k\eps^{-1})}\right)}\]
	 samples and output a hypothesis of error at most $\nmarg(\rho,\calD_{t-1},f) + \eps$. We also make at most $\gamma^{-1} \log(\eps^{-1})$ calls to \findh{} each of which takes time at most 
	 	\[\poly(\eps^{-1}, k, \rho^{-1}) \cdot 2^{O \left( \sqrt{n \log(1/\rho) \log(k\eps^{-1})} \right)}.\]
	 Notably, \Cref{lem:cover-learn} also implies that $h$ will indeed be an intersection of halfspaces.
	 
	 Thus we have proven the theorem with success probability $0.96$. By applying confidence boosting \cite{kearns1994introduction}, we can then boost the success probability to $1 - \delta$ at the cost of repeating $O(\log(1/\delta))$ times and complete the proof.
\end{proof}

\section*{Acknowledgements}
Shyamal Patel is supported by NSF grants CCF-2106429, CCF-2107187, CCF-2218677, ONR grant ONR-13533312, and an NSF Graduate Student Fellowship. Santosh Vempala is supported by NSF grant CCF-2504994 and a Simons Investigator award. The authors thank Rocco Servedio and the anonymous FOCS reviewers for helpful feedback.

\begin{flushleft}
\bibliographystyle{alpha}
\bibliography{allrefs}
\end{flushleft}

\appendix
\section{SQ Lower Bounds for Learning Halfspaces}
\label{app:SQ}

We include a simple SQ lower bound of $2^{\Omega(\sqrt{n})}$ for learning  intersections of halfspaces over $\{0,1\}^n$ that is essentially a rephrasing of \cite{tiegel2024improved}. In particular, we will prove

\begin{theorem}
\label{thm:sq}
Any SQ algorithm that learns an intersection of $\sqrt{n}$ halfspaces over $\zo^n$ either makes $2^{\Omega(\sqrt{n})}$ queries or makes a query of tolerance at most $2^{-\Omega(\sqrt{n})}$.
\end{theorem}

To prove this, we will prove a lower bound on the SQ dimension of the class. Recall that this is defined as follows. 

\begin{definition}[SQ Dimension]
	The SQ dimension of a class $\mathcal{C}$ under a distribution $\mathcal{D}$ is defined as the largest number $d$ such that there exists functions $f_1, \dots, f_d$ such that
		\[\left| \E_{\bx \sim \calD} [ f_i(x) f_j(x) ] \right| \leq \frac{1}{d} \]
	for all $i \not = j$.
	The SQ dimension of the class $\mathcal{C}$ is then the maximum over all distributions $\mathcal{D}$ of the SQ dimension with respect $\mathcal{D}$.
\end{definition}

Notably, this quantity controls the the number of queries needed to learn $\mathcal{C}$ in the statistical query model. 

\begin{theorem}[\cite{blum1994weakly}]
\label{thm:sq-dim-to-queries}
	Suppose that $\mathcal{C}$ is a concept class with SQ-dimension $d$. Then any SQ algorithm must either make $\Omega(d^{1/3})$ queries or a query of tolerance at most $O(d^{-1/3})$.
\end{theorem}

Given this, it will suffice to prove the following lemma.

\begin{lemma}
\label{lem:sq-dim}
Suppose that $n$ is an integer of the form $m^2 + m$, then for any $k \leq \frac{1}{2} m$ the SQ dimension of an intersection of halfspaces is at least $\binom{m}{\leq 2k}$.	
\end{lemma}

\begin{proof}
We prove the statement by showing how to write parities of size at most $2k$ as intersections of $k$ halfspaces.

As a starting observation, we note that we can design a degree-$2$ polynomial threshold function (PTF) that is equal to $0$ on the $\ell$th hamming level of the cube and $1$ everywhere else. Indeed, we can consider
	\[\sgn \left( \left( \ell - \sum_{i=1}^n x_i\right)^2 - 0.1\right)\]
We can then write any parity on a set of $k$ variables $S \subseteq [n]$ as an intersection of $\lfloor k/2 \rfloor$ degree-$2$ PTFs. In particular, we have that
\begin{equation}
\label{eq:parity-to-intersection}	
\chi_{S}(x) = \bigwedge_{i=0}^{\lfloor k/2 \rfloor} \sgn \left( \left( 2 \cdot i - \sum_{i \in S} x_i\right)^2 - 0.1\right).
\end{equation}

Given this, we will now define our distribution $\mathcal{D}$ and family of orthogonal functions. To start, we sample from $\mathcal{D}$ by drawing $\bx \sim \{0,1\}^{m}$ uniformly and then output $\bz = \bx^{\otimes 2} \circ \bx \in \{0,1\}^n$, where $\circ$ denotes concatenation. We then choose our orthogonal family of functions $\calF$ to correspond to $\chi_S(x)$ for some $S \subseteq [m]$ of size at most $2k$ as in \Cref{eq:parity-to-intersection}. (Note since we added a coordinate for each degree-$2$ monomial, these parities are each an intersection of halfspaces over our new features space $\{0,1\}^n$.) We then have that for any two functions $f,g \in \calF$ corresponding to parities over subsets $S \not = T$,
	\[\E_{\bz \sim \calD} [f(\bz) g(\bz)] = \E_{\bx \sim \{0,1\}^m} [ \chi_S(\bx) \chi_T(\bx)] = 0.\]
This completes the proof.
\end{proof}

With this, we can now prove \Cref{thm:sq}.

\begin{proof}[Proof of \Cref{thm:sq}]
	Let $m$ denote the largest integer such that $m^2 + m \leq n$ and set $k = \frac{1}{m}$. By \Cref{lem:sq-dim}, an intersection of $k$ halfspaces over $\{0,1\}^{m^2+m}$ has SQ dimension at least $\binom{m}{\leq 2k} = 2^{\Omega(\sqrt{n})}$. By a simple padding argument, this then implies that the same result holds for the SQ dimension of an intersection of $k$ halfspaces over $\{0,1\}^n$. Applying \Cref{thm:sq-dim-to-queries} then yields the result.
\end{proof}

\end{document}